\documentclass[preprint]{elsarticle}

\usepackage{hyperref}

\journal{Applied Mathematics Letters}

\usepackage{amsmath}
\usepackage{amssymb}
\usepackage{amsthm}
\usepackage{amsfonts}
\usepackage{calligra}
\usepackage{appendix}
\usepackage{graphicx}
\usepackage{enumerate}
\usepackage{mathrsfs}
\usepackage{mathtools}
\usepackage{color}
\usepackage{wrapfig,caption}
\usepackage{tkz-euclide}
\usepackage{relsize}

\newcommand{\be}{\begin{equation}}
\newcommand{\en}{\end{equation}}

\def\R{\mathbb R}

\def\N{\mathbb N}

\def\d{\hspace{0.15em} \textup{d}}
\def\B{\mathbf{B}}
\def\v{\mathbf{v}}
\def\A{\mathbf{A}}
\def\d{\,\textup{d}}
\def\tp{\textup}
\def\n{\mathbf{n}}
\def\t{\mathbf{t}}
\def\v{\mathbf{v}}
\def\u{\mathbf{u}}
\def\w{\mathbf{w}}
\def\0{\mathbf{0}}
\def\K{\mathcal{K}}

\def\Q{\mathcal{Q}}
\def\X{\boldsymbol{\mathcal{X}}}
\def\brho{\boldsymbol\rho}

\def\pplus{\stackrel{\perp}{\oplus}}

\newcommand{\defeq}{\mathrel{:\mkern-0.25mu=}}

\newtheorem{theorem}{Theorem}

\theoremstyle{definition}

\theoremstyle{remark}
\newtheorem{remark}[]{Remark}
\def\curl{{\rm curl}}

\begin{document}

\begin{frontmatter}

\title{On the proof of Taylor's conjecture\\ in multiply connected domains}

\author{Daniel Faraco}
\address{Departamento de Matem\'aticas, Univeridad Aut\'{o}noma de Madrid, E-28049 Madrid, Spain; ICMAT CSIC-UAM-UC3M-UCM, E-28049 Madrid, Spain}

\author{Sauli Lindberg}
\address{Department of Mathematics and Statistics, University of Helsinki, P.O. Box 68, \\00014 Helsingin yliopisto, Finland}

\author{David MacTaggart}
\address{School of Mathematics and Statistics, University of Glasgow, \\Glasgow, G12 8QQ, United Kingdom}

\author{Alberto Valli}
\address{Department of Mathematics, University of Trento, 38123 Povo (Trento), Italy}

\begin{abstract}
In this Letter we extend the proof, by Faraco and Lindberg~\cite{FL20}, of Taylor's conjecture in multiply connected domains to cover arbitrary vector potentials and remove the need to impose conditions on the magnetic field due to gauge invariance.  This extension allows us to treat general magnetic fields in closed domains that are important in laboratory plasmas and brings closure to a problem that has been open for almost 50 years. 
\end{abstract}


\end{frontmatter}


\section{Introduction}

The Woltjer--Taylor theory of magnetic relaxation is remarkable in the sense that it is one of the few examples of turbulence theory for which the end result can be written down explicitly. In short, the theory predicts that for a magnetically-closed plasma, a magnetic field will evolve through turbulent relaxation to the equilibrium of a linear force-free field. This relaxation is based on two key results: Woltjer's theorem \cite{W58} and Taylor's conjecture \cite{T74,T86}. Woltjer's theorem states that in a fixed and magnetically-closed domain, the magnetic field that minimizes the magnetic energy whilst keeping the global magnetic helicity fixed is a linear force-free field. Taylor combined this theorem with the conjecture that in a non-ideal plasma (one with small but non-zero resistivity) the global helicity of the entire domain is approximately conserved. Related to this, Taylor also assumed that the magnetic helicities of subdomains (sub-helicities) are not important in determining the final relaxed state since the magnetic field undergoes significant field line reconnection during turbulent relaxation, thus destroying the sub-helicities.

Despite some good experimental evidence for Woltjer--Taylor relaxation, see \cite{BPV76,T86,JPS95}, there also exists research which suggests that the assumption that only the global magnetic helicity is conserved is too restrictive in general. Simulations, such as \cite{YRH15}, show that magnetic fields do not always relax to linear force-free fields. This deviation from the classic Woltjer--Taylor theory has been attributed to the influence of other topological invariants besides the global magnetic helicity \cite{YH11}. It has also been suggested that field line helicities (sub-helicities of individual field lines) can play an important role in magnetic relaxation, and that they are redistributed by reconnection rather than destroyed. See \cite{Y20} for a summary of these issues. 

Although modifications to the classic Woltjer--Taylor theory seem to be needed for a more comprehensive theory of magnetic relaxation, one aspect of the theory that has proved to be robust is Taylor's conjecture itself. Berger \cite{B84} showed that strict limits on magnetic helicity dissipation can be found, which depend on the plasma energy, the energy dissipation rate and the magnetic diffusion. In both astrophysical and laboratory plasmas, the conditions are often such that magnetic helicity dissipation, following Berger's bounds, is very small and so magnetic helicity can be treated as an invariant in non-ideal (resistive) plasmas as well as ideal plasmas. For a list of other relevant theoretical works on helicity conservation, we direct the reader to \cite[][Sect. 1]{FL20}.

Recently, Faraco and Lindberg~\cite{FL20} provided a rigorous proof of Taylor's conjecture in a magnetically-closed turbulent plasma. They described the turbulent plasma using Onsager's theory of turbulence \cite{O49,ES06} and proved the conjecture by considering limits of Leray--Hopf solutions of non-ideal magnetohydrodynamics (MHD). The proof  considers both simply connected and multiply connected domains. In the former, the choice of vector potential in the magnetic helicity is arbitrary, but in the latter this is not the case. The fact that only certain vector potentials are suitable in multiply connected domains stems from~\cite{FL20} considering the classical form of magnetic helicity,
\be\label{hel_class}
H = \int_{\Omega}\A\cdot\B,
\en 
where $\B$ is the magnetic field, $\A$ is the vector potential with $\tp{\curl}\,\A=\B$ and $\Omega$ is the magnetically-closed domain. Equation (\ref{hel_class}) represents the gauge invariant form of helicity in simply connected domains. MacTaggart and Valli~\cite{MV19} showed that the definition of magnetic helicity $H$ needs to be extended for multiply connected domains in order to accommodate arbitrary vector potentials that do not impose any restrictions on the magnetic field. For example, it can be shown that a consequence of imposing gauge invariance with equation (\ref{hel_class}) in multiply connected domains is that the magnetic flux through internal cuts of the domain must be zero, e.g., the toroidal magnetic flux in a tokamak would be zero.

In this Letter, we bring together the results of Faraco and Lindberg~\cite{FL20}  and MacTaggart and Valli~\cite{MV19} to provide a complete proof of Taylor's conjecture for arbitray vector potentials of closed magnetic fields with no restrictions on the magnetic field imposed by gauge invariance. We proceed by only describing extensions and modifications to specific parts of the proof in~\cite{FL20} that are relevant for expanding its scope in the manner described above. For full details of the many technical results involved in the original proof, the reader is directed to the work by Faraco and Lindberg~\cite{FL20}. Before presenting the main result, we now introduce some preliminary material that will be key to understanding the extension to their proof.

\section{Preliminaries}
\subsection{Geometric set-up}
We consider a multiply connected domain $\Omega\subset\mathbb{R}^3$ that is a bounded open connected set with Lipschitz continuous boundary $\partial\Omega$ and unit outer normal $\n$. Let $g>0$ be the first Betti number, or genus, of $\Omega$. Then the first Betti number of $\partial\Omega$ is 2$g$. We can consider 2$g$ non-bounding cycles on $\partial\Omega$, $\{\gamma_i\}_{i=1}^g\cup\{\gamma'_i\}_{i=1}^g$, that represent the generators of the first homology group of $\partial\Omega$. The $\{\gamma_i\}_{i=1}^g$ represent the generators of the first homology group of $\overline{\Omega}$. The tangent vector on $\gamma_i$ is denoted $\t_i$. Similarly, the $\{\gamma'_i\}_{i=1}^g$ represent the generators of the first homology group of $\overline{\Omega'}$, where $\Omega'=B\setminus\overline{\Omega}$ and $B$ is an open ball containing $\overline{\Omega}$. The tangent vector on $\gamma'_i$ is denoted $\t'_i$.

In $\Omega$ there exist $g$ `cutting surfaces' $\{\Sigma_i\}_{i=1}^g$ that are connected orientable Lipschitz surfaces satisfying $\Sigma_i\subset\Omega$. Each cutting surface in $\Omega$ satisfies $\partial\Sigma_i=\gamma'_i$ and $\overline{\Sigma}_i$ `cuts' the corresponding cycle $\gamma_i$ in exactly one point. Each cutting surface $\Sigma_i$ has a unit normal vector $\n_{\Sigma_i}$ oriented as $\gamma_i$. An equivalent definition applies for the cutting surfaces $\{\Sigma'_i\}_{i=1}^g$ of $\Omega'$. In particular, $\partial \Sigma'_i = \gamma_i$. 

Two examples of the geometric set-up are displayed in Figure \ref{geo_fig1} for a 2-fold torus and in Figure \ref{geo_fig2} for a toroidal shell.

\begin{figure}[h!]
\centering
{\begin{tikzpicture}[scale=0.6]
\draw[thick] (-6,0)--(0,0);
\draw[thick] (-6,1)--(0,1);
\draw[thick] (-6,0)--(-6,1);
\draw[thick] (0,0)--(0,1);
\draw[thick] (0,0)--(5,3);
\draw[thick] (0,1)--(5,4);
\draw[thick] (5,3)--(5,4);
\draw[thick] (-6,1)--(-1,4);
\draw[thick] (5,4)--(-1,4);
\draw[thick] (5,4)--(-1,4);

\draw[thick,blue] (-4,1.5)--(-2.25,1.5);
\draw[thick,blue,->] (-0.5,1.5)--(-2.25,1.5);
\draw[thick,blue] (-4,1.5)--(-4+0.8*5/3,1.5+0.8*1); 
\draw[thick,blue] (-0.5,1.5)--(-0.5+0.8*5/3,1.5+0.8*1); 
\draw[thick,blue] (-4+0.8*5/3,1.5+0.8*1)--(-0.5+0.8*5/3,1.5+0.8*1);
\node[blue] at (-3.65,2.4) [anchor=north]{\small{$\gamma_1$}};

\draw[thick] (-4+0.8*5/3,1.5+0.8*1)--(-4+0.8*5/3,1.5);

\fill[orange,fill opacity=0.1]  (-4,1.5)--(-0.5,1.5)--(-0.5+0.8*5/3,1.5+0.8*1)--(-4+0.8*5/3,1.5+0.8*1);
\node at (-1.5,2.3) [anchor=north]{\small{$\Sigma'_1$}};

\draw[thick,blue] (-4+1.3*5/3,1.5+1.3*1)--(-2.25+1.3*5/3,1.5+1.3*1);
\draw[thick,blue] (-4+2.1*5/3,1.5+2.1*1)--(-0.5+2.1*5/3,1.5+2.1*1);
\draw[thick,blue] (-4+1.3*5/3,1.5+1.3*1)--(-4+2.1*5/3,1.5+2.1*1); 
\draw[thick,blue] (-0.5+1.3*5/3,1.5+1.3*1)--(-0.5+2.1*5/3,1.5+2.1*1); 
\draw[thick,blue,->] (-0.5+1.3*5/3,1.5+1.3*1)--(-2.25+1.3*5/3,1.5+1.3*1);
\node[blue] at (-1.61,3.65) [anchor=north]{\small{$\gamma_2$}};

\draw[thick] (-4+2.1*5/3,1.5+2.1*1)--(-4+2.1*5/3,1.5+1.3*1);

\fill[orange,fill opacity=0.1]  (-4+1.3*5/3,1.5+1.3*1)--(-0.5+1.3*5/3,1.5+1.3*1)--(-0.5+2.1*5/3,1.5+2.1*1)--(-4+2.1*5/3,1.5+2.1*1);
\node at (0.6,3.7) [anchor=north]{\small{$\Sigma'_2$}};


\draw[thick,red] (-0.5+0.4*5/3,1.5+0.4*1)--(0.15+0.4*5/3,1.5+0.4*1);
\draw[thick,red,->] (0.8+0.4*5/3,1.5+0.4*1)--(0.15+0.4*5/3,1.5+0.4*1);
\draw[thick,red] (0.8+0.4*5/3,1.5+0.4*1-1)--(0.8+0.4*5/3,1.5+0.4*1);
\draw[thick,red,dashed]  (-0.5+0.4*5/3,1.5+0.4*1)-- (-0.5+0.4*5/3,1.5+0.4*1-1);
\draw[thick,red,dashed]  (-0.5+0.4*5/3,1.5+0.4*1-1)--(0.8+0.4*5/3,1.5+0.4*1-1);

\fill[cyan,fill opacity=0.1]  (-0.5+0.4*5/3,1.5+0.4*1)--(0.8+0.4*5/3,1.5+0.4*1)--(0.8+0.4*5/3,1.5+0.4*1-1)--(-0.5+0.4*5/3,1.5+0.4*1-1);
\node at (1,1.65) [anchor=north]{\small{$\Sigma_1$}};
\node[red] at (1.8,1.95) [anchor=north]{\small{$\gamma'_1$}};

\draw[thick,red] (-0.5+1.65*5/3,1.5+1.65*1)--(0.155+1.65*5/3,1.5+1.65*1);
\draw[thick,red,->] (0.81+1.65*5/3,1.5+1.65*1)--(0.155+1.65*5/3,1.5+1.65*1);
\draw[thick,red,]  (0.81+1.65*5/3,1.5+1.65*1-1)--(0.81+1.65*5/3,1.5+1.65*1);
\draw[thick,red,dashed]  (-0.5+1.65*5/3,1.5+1.65*1)--(-0.5+1.65*5/3,1.5+1.65*1-1);
\draw[thick,red,dashed] (-0.5+1.65*5/3,1.5+1.65*1-1)-- (0.81+1.65*5/3,1.5+1.65*1-1);

\fill[cyan,fill opacity=0.1] (-0.5+1.65*5/3,1.5+1.65*1-1)-- (0.81+1.65*5/3,1.5+1.65*1-1)--(0.81+1.65*5/3,1.5+1.65*1)--(-0.5+1.65*5/3,1.5+1.65*1);
\node at (3.1,2.85) [anchor=north]{\small{$\Sigma_2$}};
\node[red] at (3.85,3.15) [anchor=north]{\small{$\gamma'_2$}};

\end{tikzpicture}}
\caption{An illustration of cutting surfaces (shaded areas) and cycles (coloured paths) for a 2-fold torus. The notation is as described in the main text.}\label{geo_fig1}
\end{figure}
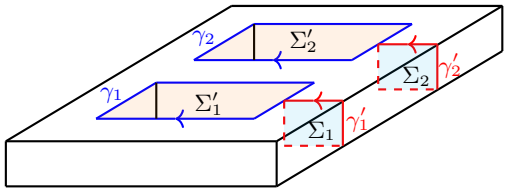

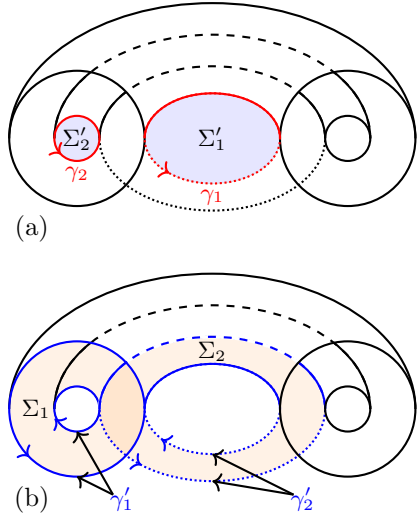
\begin{figure}[h!]
\centering
{\begin{tikzpicture}[scale=0.6]


\draw [fill=blue, opacity=0.1] (0,0) ellipse (1.5cm and 1cm);
\draw [thick] (3,0) ellipse (1.5cm and 1.5cm);
\draw [thick] (-3,0) ellipse (1.5cm and 1.5cm);
\draw [thick] (3,0) ellipse (0.5cm and 0.5cm);
\draw [fill=blue, opacity=0.1] (-3,0) ellipse (0.5cm and 0.5cm);

\draw[black, thick]  (-3.5,0) arc(-180:-220:3.5cm and 2.3cm);
\draw[black, thick]  (3.5,0) arc(-360:-320:3.5cm and 2.3cm);
\draw[black, thick,dashed]  (-3.5,0) arc(-180:-360:3.5cm and 2.3cm);

\draw[black, thick,dashed]  (-2.5,0) arc(-180:-360:2.5cm and 1.6cm);
\draw[black, thick]  (-2.5,0) arc(-180:-220:2.5cm and 1.6cm);
\draw[black, thick]  (2.5,0) arc(-360:-320:2.5cm and 1.6cm);

\draw[black, thick]  (-4.5,0) arc(-180:-360:4.5cm and 3cm);
\draw[red, thick]  (-1.5,0) arc(-180:-360:1.5cm and 1.cm);

\draw[black, thick,densely dotted]  (-2.5,0) arc(-180:0:2.5cm and 1.6cm);

 \draw[
        red,thick,decoration={markings, mark=at position 0.625 with {\arrow{>}}},
        postaction={decorate}
        ]
        (-3,0) ellipse (0.5cm and 0.5cm);

 \draw[
        red,thick,densely dotted,decoration={markings, mark=at position 0.625 with {\arrow{>}}},
        postaction={decorate}
        ]
         (0,0) ellipse (1.5cm and 1cm);

\node at (-3,0) {\small{$\Sigma'_2$}};
\node at (0,0) {\small{$\Sigma'_1$}};
\node[red] at (0,-1.3) {\small{$\gamma_1$}};
\node[red] at (-3,-0.8) {\small{$\gamma_2$}};


\draw [fill=orange, opacity=0.1] (-3,-6) ellipse (1.5cm and 1.5cm);

\draw [fill=white] (-3,-6) ellipse (0.5cm and 0.5cm);

\draw[black, thick]  (-3.5,-6) arc(-180:-220:3.5cm and 2.3cm);
\draw[black, thick]  (3.5,-6) arc(-360:-320:3.5cm and 2.3cm);
\draw[black, thick,dashed]  (-3.5,-6) arc(-180:-360:3.5cm and 2.3cm);

\draw[black, thick,dashed]  (-2.5,-6) arc(-180:-360:2.5cm and 1.6cm);
\draw[black, thick]  (-2.5,-6) arc(-180:-220:2.5cm and 1.6cm);
\draw[black, thick]  (2.5,-6) arc(-360:-320:2.5cm and 1.6cm);

\draw[black, thick]  (-4.5,-6) arc(-180:-360:4.5cm and 3cm);

\draw[draw=none,fill=orange,opacity=0.1]  (-2.5,-6) arc(-180:180:2.5cm and 1.6cm);
\draw [fill=white,draw=none] (0,-6) ellipse (1.5cm and 1cm);

 \draw[
        blue,thick,decoration={markings, mark=at position 0.625 with {\arrow{<}}},
        postaction={decorate}
        ]
        (-3,-6) ellipse (0.5cm and 0.5cm);
        
 \draw[
        blue,thick,densely dotted,decoration={markings, mark=at position 0.625 with {\arrow{<}}},
        postaction={decorate}
        ]
        (0,-6) ellipse (1.5cm and 1cm);

\draw[
        blue,thick,densely dotted,decoration={markings, mark=at position 0.625 with {\arrow{>}}},
        postaction={decorate}
        ]
        (0,-6) ellipse (2.5cm and 1.6cm);

\draw[white, very thick]  (-2.5,-6) arc(-180:-360:2.5cm and 1.6cm);
\draw[blue, thick,dashed]  (-2.5,-6) arc(-180:-360:2.5cm and 1.6cm);
\draw[blue, thick]  (-2.5,-6) arc(-180:-221:2.5cm and 1.6cm);
\draw[blue, thick]  (2.5,-6) arc(-360:-320:2.5cm and 1.6cm);
\draw[blue, thick]  (-1.5,-6) arc(-180:-360:1.5cm and 1cm);

\draw[
        blue,thick,decoration={markings, mark=at position 0.625 with {\arrow{>}}},
        postaction={decorate}
        ]
        (-3,-6) ellipse (1.5cm and 1.5cm);
\draw [thick] (3,-6) ellipse (1.5cm and 1.5cm);
\draw [thick] (3,-6) ellipse (0.5cm and 0.5cm);

\node at (0,-4.725) {\small{$\Sigma_2$}};
\node at (-3.9,-6) {\small{$\Sigma_1$}};

\node[blue] at (-2,-8) {\small{$\gamma'_1$}};
\draw[thick,->] (-2.2,-7.9) -- (-3,-7.5);
\draw[thick,->] (-2.2,-7.9) -- (-3,-6.5);

\node[blue] at (2,-8) {\small{$\gamma'_2$}};
\draw[thick,->] (1.8,-7.9) -- (0,-7.6);
\draw[thick,->] (1.8,-7.9) -- (0,-7);

\node at (-4,-2) {(a)};
\node at (-4,-8) {(b)};

\end{tikzpicture}}
\caption{An illustration of cutting surfaces (shaded areas) and cycles (coloured paths) for a toroidal shell (cut in half). (a) shows the cutting surfaces of $\Omega'$ and their associated cycles. (b) shows the same but for $\Omega$.}\label{geo_fig2}
\end{figure}

\subsection{Required function spaces}
Using the basic Hilbert spaces listed in Alonso Rodr{\'\i}guez et al.~\cite[][Sect. 2]{ACRVV18}, the space of harmonic Neumann fields is defined by
\[\K_T \defeq H(\tp{curl}^0;\Omega) \cap H_0(\tp{div}^0;\Omega) .
\]
Remember that if $\v \in H_0(\tp{div}^0;\Omega)$ satisfies $\v \perp \K_T$ (where orthogonality has to be intended in $(L^2(\Omega))^3$) then it follows that $\int_{\Sigma_i} \v \cdot \n_{\Sigma_i} \d S = 0$ for $i = 1,\ldots,g$, and viceversa. Thus, if we denote
\[\Q \defeq \nabla H^1(\Omega) \pplus \K_T,\]
we have
\begin{align*}
& \Q = H(\tp{curl}^0;\Omega), \\
& \Q^\perp = \left\{ \v \in H_0(\tp{div}^0;\Omega) \colon \int_{\Sigma_i} \v \cdot \n_{\Sigma_i} = 0 \tp{ for } i = 1,\ldots,g \right\}, \\
& H_0(\tp{div}^0;\Omega) = \Q^\perp \pplus \K_T.
\end{align*}
Analogous spaces can be defined for $\Omega'$, e.g., $\K_T' \defeq H(\tp{curl}^0;\Omega') \cap H_0(\tp{div}^0;\Omega')$. For the ease of the reader, let us also recall that the notation in Faraco and Lindberg~\cite{FL20} is different: precisely, $L^2_\sigma(\Omega;\R^3) = H_0(\tp{div}^0;\Omega)$, $L^2_H(\Omega;\R^3) = \K_T$, $L^2_\Sigma(\Omega;\R^3) = \Q^\perp$.

\subsection{Integral identities}

Let the basis functions of $\K_T$ and $\K_T'$ be denoted by $\brho_i$ and $\brho'_i$, $i=1\dots g$, respectively. These functions are defined by imposing that $\oint_{\gamma_k}\brho_i\cdot\t_k = \delta_{ki}$ and $\oint_{\gamma'_k}\brho'_i\cdot\t'_k = \delta_{ki}$.

Let us recall that, given $\v\in H_0(\tp{div}^0;\Omega)$, it holds that
\be\label{fluxSi}
\int_{\Omega}\v\cdot\brho_i = \int_{\Sigma_i}\v\cdot\n_{\Sigma_i}. 
\en 

Given $\w\in \X \defeq \{\w\in H(\tp{curl};\Omega): \curl\,\w\cdot\n = 0\,\tp{on}\,\partial\Omega\}$, as in Alonso Rodr{\'\i}guez et al.~\cite{ACRVV18} we define
\[
\oint_{\gamma_i}\w\cdot\t_i \defeq -\int_{\partial\Omega}\n\times\w\cdot\brho'_i, \quad \oint_{\gamma'_i}\w\cdot\t'_i \defeq \int_{\partial\Omega}\n\times\w\cdot\brho_i
\]
(as in~\cite{ACRVV18} the generators of the first homology group of $\overline{\Omega}$ are denoted $\gamma'_i$ and those of the first homology group of $\overline{\Omega'}$ are denoted $\gamma_i$). Note that, with these definitions and taking into account \eqref{fluxSi}, Stokes' theorem holds on $\Sigma_i$, i.e., $\oint_{\gamma'_i}\w\cdot\t'_i = \int_{\Sigma_i} \curl\,\w\cdot\n_{\Sigma_i}$.

As in Lemma 7 of~\cite{ACRVV18}, for $\w$, $\v\in\X$ we find
\be\label{iden}
\begin{array}{ll}
\int_{\partial\Omega}(\w\times\n)\cdot\v \!\!\!&= \sum_{i=1}^g\left(\oint_{\gamma_i}\w\cdot\t_i\right)\left(\oint_{\gamma'_i}\v\cdot\t'_i\right) \\ &\qquad - \sum_{i=1}^g\left(\oint_{\gamma'_i}\w\cdot\t'_i\right)\left(\oint_{\gamma_i}\v\cdot\t_i\right).
\end{array}
\en
For more details on the derivation of this result, the reader is directed to~\cite{ACRVV18}.

\subsection{Leray--Hopf solutions of viscous, resistive MHD}
In ~\cite{FL20}, the domains studied have a $C^{1,1}$ regular boundary. This smoothness condition implies, in particular, that $X_T(\Omega) \defeq H(\tp{curl};\Omega) \cap H_0(\tp{div};\Omega) \hookrightarrow H^1(\Omega) \hookrightarrow L^6(\Omega)$, and this chain of embeddings allows one to prove the existence of Leray--Hopf solutions for the velocity $\u \in L^\infty(0,T;H_0(\tp{div}^0(\Omega))) \cap L^2(0,T; H_0^1(\Omega))$ and magnetic field $\B \in L^\infty(0,T;H_0(\tp{div}^0(\Omega))) \cap L^2(0,T; H^1(\Omega))$ by standard arguments (see~\cite[Appendix A]{FL20}).

In fact, Lipschitz regularity of $\partial \Omega$ is also a strong enough assumption for the existence of Leray--Hopf solutions, as long as the definition of Leray--Hopf solutions is modified in the manner described below. First recall that in bounded Lipschitz domains, the weaker embedding result $X_T(\Omega) \hookrightarrow H^{1/2}(\Omega) \hookrightarrow L^3(\Omega)$ holds. We relax the smoothness requirements on the magnetic field to $\B \in L^\infty(0,T;H_0(\tp{div}^0(\Omega))) \cap L^2(0,T; X_T(\Omega))$. Whereas in $C^{1,1}$ regular domains the weak formulation of the Cauchy momentum equation says that, for viscosity $\nu$, $\langle \partial_t \u, \varphi \rangle + \int_\Omega (\u \cdot \nabla \u - \B \cdot \nabla \B) \cdot \varphi + \nu \int_\Omega \nabla \u : \nabla \varphi = 0$ holds at almost each $t \in [0,T)$ for every $\varphi \in H(\tp{div}^0(\Omega)) \cap H^1_0(\Omega)$, in bounded Lipschitz domains the formula $\langle \partial_t \u, \varphi \rangle + \int_\Omega (\u \cdot \nabla \u - \tp{curl} \, \B \times \B) \cdot \varphi + \nu \int_\Omega \nabla \u : \nabla \varphi = 0$ needs to be used instead; note that $\|\tp{curl} \, \B \times \B \cdot \varphi\|_{L^1} \le \|\tp{curl} \, \B\|_{L^2} \|\B\|_{L^3} \|\varphi\|_{L^6}$ at almost each $t \in [0,T)$ by H\"{o}lder's inequality. Under these modifications, the main result~\cite[Theorem 1.5]{FL20} holds and the proof only requires minor adjustments.

\section{Main result}
Suppose that $\B,\u \in L^\infty(0,T;H_0(\tp{div}^0;\Omega))$ form a weak solution of ideal MHD with initial data $\B_0,\u_0 \in H_0(\tp{div}^0;\Omega)$, where $\u$ is the velocity. Assume, furthermore, that $(\B,\u)$ arises as an ideal (inviscid, non-resistive) limit of Leray--Hopf solutions. Using the notation from~\cite{FL20} we denote
\be\label{decomp_b}
\B = \B^\Sigma + \B^H, \qquad \B^\Sigma \in \Q^\perp \tp{ and } \B^H \in \K_T.
\en
Let us recall that  at almost each time $\B^H$ is the unique element of $\K_T$ such that $\int_{\Sigma_i} \B^H \cdot \n_{\Sigma_i} = \int_{\Sigma_i} \B \cdot \n_{\Sigma_i}$ for each $i=1,\ldots,g$.
Further, we decompose the vector potential as
\be\label{decomp_a}
\A = \A^\Sigma + \A^H,
\en
where at almost each time $\A^H$ is the unique element of $\Q^\perp$ such that $\tp{curl} \, \A^H = \B^H$ (for these existence results see, e.g.,~\cite[Sect.\ 2.2,  and in particular Remark 1]{ABV19}). Similar notation is used for $\B_0$ and $\A_0$.

The generalized  expression for magnetic helicity in multiply connected domains is defined as
\be\label{hel_gen}
\Upsilon(\B) \defeq \int_{\Omega}\A\cdot\B - \sum_{i=1}^g\oint_{\gamma_i}\A\cdot\t_i\int_{\Sigma_i}\B\cdot\n_{\Sigma_i}
\en
(see MacTaggart and Valli~\cite{MV19}). It is important to note that this definition is gauge invariant, i.e., it is independent of the vector potential $\A \in \mathcal{X}$. 

Using this gauge invariant expression for magnetic helicity we are now in a position to write an adapted version of Theorem 1.5 of Faraco and Lindberg~\cite{FL20}, which describes the convergence of magnetic helicity; clearly, no dependence on the vector potential appears.

\begin{theorem}
In a domain $\Omega\subset\mathbb{R}^3$ with the properties described previously, suppose that $(\u,\B)$ is a weak ideal limit of Leray--Hopf solutions $(\u_k,\B_k)$, $k\in\mathbb{N}$. The initial data of $\B$ and $\B_k$ are denoted $\B_0$ and $\B_{k,0}$, respectively. Then it holds that
\be\label{diss}
\Upsilon(\B_k,t)-\Upsilon(\B_{k,0}) = -2\eta_k\int_0^t\int_{\Omega} \B_k\cdot\curl\,\B_k,
\en 
where $\eta_k > 0$ is the resistivity, for all $k\in\mathbb{N}$ and $t\in[0,T)$. Furthermore
\be\label{hel_con}
\Upsilon(\B,t)-\Upsilon(\B_0) = 0,
\en 
for almost each $t\in[0,T)$.
\end{theorem}
\begin{proof}
Let us define
$$
Z(\B_k,\A_k;t) \defeq \int_\Omega \A_k(t) \cdot \B_k(t) - \int_{\partial \Omega} \A_k^\Sigma(t) \times \n \cdot  \A_k^H(t)
$$
and similarly for $\B_{k,0}$ and $\A_{k,0}$, $\B$ and $\A$, $\B_{0}$ and $\A_{0}$. In~\cite[Lemma 4.3]{FL20} it is shown that for Leray--Hopf solutions (and their weak limits) the harmonic part of the magnetic field is stationary, i.e.,
\[
\B_k^H(t) = \B^H_{k,0}, \quad \B^H(t) = \B^H_{0},
\]
and consequently $\A^H_k(t)=\A^H_{k,0}$, $\A^H(t)=\A^H_{0}$. Therefore, Theorem 1.5 in~\cite{FL20} (having corrected for a minus sign that is a misprint)  says that
 \begin{subequations}
 \begin{align}
& Z(\B_k,\A_k;t) - Z(\B_{k,0},\A_{k,0}) = - 2 \eta_k \int_0^t\int_{\Omega} \B_k\cdot\curl\,\B_k, \label{difference_in_Zk}\\
& Z(\B,\A;t) - Z(\B_{0},\A_0) = 0. \label{difference_in_Z}
\end{align}
\end{subequations}
Let us compute $\int_{\partial \Omega} \A_k^\Sigma(t) \times \n \cdot \A_k^H(t)$. By \eqref{iden} we have
\begin{align*}
\int_{\partial \Omega} \A_k^\Sigma(t) \times \n \cdot \A_k^H(t) = & \sum_{i=1}^g\left(\oint_{\gamma_i}\A^\Sigma_k(t)\cdot\t_i\right)\left(\oint_{\gamma'_i}\A_{k,0}^H\cdot\t'_i\right) \\
& -\sum_{i=1}^g\left(\oint_{\gamma'_i}\A^\Sigma_k(t)\cdot\t'_i\right)\left(\oint_{\gamma_i}\A_{k,0}^H\cdot\t_i\right). 
\end{align*}
Using the decomposition in (\ref{decomp_b}), we know that 
\be\label{diff}
\int_{\Sigma_i}\B^\Sigma_k(t) \cdot\n_{\Sigma_i} = \int_{\Sigma_i}\B_{k,0}^\Sigma\cdot\n_{\Sigma_i} = 0.
\en
Since $\curl\,\A_{k}^\Sigma(t)=\B_{k}^\Sigma(t)$, $\curl\,\A_{k,0}^H=\B_{k,0}^H$ and $\gamma'_i=\partial\Sigma_i$, applying Stokes' theorem gives
\[
\begin{array}{ll}
&\oint_{\gamma'_i}\A_{k}^\Sigma(t) \cdot\t'_i=\int_{\Sigma_i}\B_{k}^\Sigma(t) \cdot\n_{\Sigma_i} = 0,\\ [5pt]
&\oint_{\gamma'_i}\A_{k,0}^H\cdot\t'_i=\int_{\Sigma_i}\B_{k,0}^H\cdot\n_{\Sigma_i} =\int_{\Sigma_i}\B_{k,0}\cdot\n_{\Sigma_i}, 
\end{array}
\]
and also 
\be\label{conservation}
\begin{array}{ll}
&\int_{\Sigma_i}\B_{k}(t) \cdot\n_{\Sigma_i} = \int_{\Sigma_i}\B_{k}^H(t) \cdot\n_{\Sigma_i} \\ [5pt]&\qquad = \oint_{\gamma'_i}\A_{k}^H(t)\cdot\t'_i = \oint_{\gamma'_i}\A_{k,0}^H\cdot\t'_i = \int_{\Sigma_i}\B_{k,0}\cdot\n_{\Sigma_i} ,
\end{array}
\en
for $t\in[0,T)$, $i=1,\ldots,g$, $k \in \N$. In conclusion, repeating the same computation for $\int_{\partial \Omega} \A_{k,0}^\Sigma \times \n \cdot  \A_{k,0}^H$, we find 
$$
\begin{array}{ll}
&Z(\B_k,\A_k;t) - Z(\B_{k,0},\A_{k,0})\\ [5pt]
&\qquad = \int_\Omega \A_k(t) \cdot \B_k(t) -  \int_\Omega \A_{k,0} \cdot \B_{k,0}\\[5pt]
&\qquad\qquad - \sum_{i=1}^g\left(\oint_{\gamma_i}[\A^\Sigma_k(t)-\A^\Sigma_{k,0}]\cdot\t_i\right)\left(\int_{\Sigma_i}\B_{k,0}\cdot\n_{\Sigma_i}\right) \\ [5pt]
&\qquad= \int_\Omega \A_k(t) \cdot \B_k(t) -  \int_\Omega \A_{k,0} \cdot \B_{k,0}\\[5pt]
&\qquad\qquad - \sum_{i=1}^g\left(\oint_{\gamma_i}[\A_k(t)-\A_{k,0}]\cdot\t_i\right)\left(\int_{\Sigma_i}\B_{k,0}\cdot\n_{\Sigma_i}\right) \\ [5pt]
&\qquad=  \Upsilon(\B_k,t)-\Upsilon(\B_{k,0}) ,
\end{array}
$$
having used  again that $\A^H_k(t)=\A^H_{k,0}$ and moreover \eqref{conservation}. By combining the computation above with \eqref{difference_in_Zk} we get \eqref{diss}.

By analogous reasoning, $Z(\B,\A;t) - Z(\B_0,\A_0) = \Upsilon(\B,t)-\Upsilon(\B_0)$. Now \eqref{difference_in_Z} leads to \eqref{hel_con}, and Theorem 1 is proved.
\end{proof} 

\begin{remark}
The right-hand side of \eqref{diss} is bounded by $3\sqrt{t \eta_k} (\|\u_{k,0}\|_{L^2}^2 + \|\B_{k,0}\|_{L^2}^2)$ for all $k \in \N$ and $t \in [0,T)$. This is seen by using Young's inequality on $\B_k/\sqrt{t}$ and $\sqrt{t} \, \tp{curl} \, \B_k$ and then appealing to the energy inequality; see~\cite[Lemma 4.6]{FL20} where Young's inequality is used on $\B_k$ and $\tp{curl}\, \B_k$ instead.
\end{remark}

\section{Conclusion}
This extension to the proof in Faraco and Lindberg~\cite{FL20} of Taylor's conjecture allows for the extra consideration of defining magnetic helicity in multiply connected domains with arbitrary vector potentials, rather than a specific subset which may impose unwanted conditions on the magnetic field, e.g., forcing zero flux through cutting surfaces. The result gives a complete answer to a problem that has been posed almost 50 years ago.

\bigskip

\noindent {\bf Acknowledgments.}
DF acknowledges financial support from the Spanish Ministry of Science and Innovation through the Severo Ochoa Programme for Centres of Excellence in R\&D(CEX2019-000904-S and by the MTM2017-85934-C3-2-P2. He is also partially supported by CAM through the Line of excellence for University Teaching Staff between CM and UAM. DF and SL are partially supported by the ERC Advanced Grant 834728. DM acknowledges support from AFOSR: grant number FA8655-20-1-7032. AV acknowledges the support of INdAM-GNCS.


\section{References}

\small{\bibliography{mybibfile}}

\end{document}